\documentclass[a4paper,10pt]{article}

\usepackage{amsthm}
\usepackage{amsmath}
\usepackage{amssymb}
\usepackage{latexsym}

\newtheorem{theorem}{Theorem}[section]
\newtheorem{proposition}[theorem]{Proposition}
\newtheorem{lemma}[theorem]{Lemma}
\newtheorem{corollary}[theorem]{Corollary}
\newtheorem{remark}[theorem]{Remark}

\newtheorem{assumption}[theorem]{Assumption}

\begin{document}

\title{On optimal strategies for utility maximizers\\ in the Arbitrage Pricing Model}

\author{Mikl\'os R\'asonyi\thanks{Alfr\'ed R\'enyi Institute of
Mathematics, Hungarian Academy of Sciences, Budapest. 
E-mail: \texttt{rasonyi@renyi.mta.hu}}}

\date{\today}

\maketitle

\begin{abstract}
We consider a popular model of microeconomics with countably many assets: the Arbitrage Pricing 
Model. We study the problem of optimal investment under an expected utility criterion and
look for conditions ensuring the existence of optimal strategies.
Previous results required a certain restrictive hypothesis on the tails 
of asset return distributions. Using a different method,  we manage to remove this hypothesis, at the price 
of stronger assumptions on the moments of asset returns. 
\end{abstract}

%\textsl{OR/MS subject classification:} Primary: investment, infinite dimensional; secondary: 
%utility/value theory, finance, securities.\\

\section{Introduction}

The Arbitrage Pricing Model was first proposed in Ross (1976). The related mathematics
was clarified by Huberman (1982). This model has become standard
material in courses of economic theory, Huang \& Litzenberger (1988), and also inspired far-reaching 
developments in financial mathematics, the theory of so-called ``large financial markets'': see  
e.g. Kabanov \& Kramkov (1994,1998), Klein \& Schachermayer (1996a,1996b), Klein (2000), Rokhlin (2008)
and Cuchiero et al. (2015).
Large financial markets constitute a suitable framework for analysing markets with many assets such 
as bond markets.

The Arbitrage Pricing Model is a one-period market model containing a countably infinite number 
of assets.
By asymptotic considerations, it makes possible the study of diversification effects in large portfolios.
Excluding various kinds of asymptotic arbitrage opportunities leads to qualitative conclusions 
on the parameters (see Assumption \ref{b} below as well as Kabanov \& Kramkov(1998)) which allow for empirical tests 
and can even be used for asset pricing purposes, see Ross (1976), Huberman (1982),
Huang \& Litzenberger (1988).

In the present paper we treat the problem of optimal investment with a criterion of expected utility.
Continuing the work initiated in R\'asonyi (2016), we seek general conditions under which the
existence of optimal portfolios can be asserted for the Arbitrage Pricing Model. 

Existence theorems in the context of finitely
many assets have been a focus of intensive study, see Kramkov \& Schachermayer (1999),
Kabanov \& Stricker (2002), Schachermayer (2001), Biagini \& Frittelli (2005),
Owen and \v{Z}itkovi\'c (2009), just to
mention a few. Knowing that there is an optimiser brings multiple benefits: it is a reassuring
fact \emph{per se}, since explicit formulas can rarely be found; it helps identifying the
conditions on the underlying market and on investor preferences that are necessary for a
well-posed problem; finally, it usually constitutes a basis for eventual numerical methods 
to find the optimiser. 

In the context of large financial markets, De Donno et al. (2005) investigated the existence of
optimisers in utility maximization for the first time. They worked with the family  of ``generalized
strategies'' which, roughly speaking, correspond to the closure (in a suitable topology) of the 
payoffs of investments into finitely many assets. Thus their optimizer lacks an interpretation
in terms of a portfolio in infinitely many assets. Part of the mathematical difficulties
related to this problem come from the fact that it is not easy to characterize this closure.
The first main contribution of the present paper is to provide reasonable conditions
under which a simple characterization is possible, see Lemma \ref{tartu} below. 
Such results were first obtained in the companion paper
R\'asonyi (2016), but under restrictive hypotheses which required, in particular, that
asset returns take arbitrarily large positive and negative values. In this article
we manage to remove this hypothesis on the tails of asset returns at the price of
imposing more stringent integrability conditions on them. See Section \ref{gorr}
for further details.

Our second main contribution is establishing that an optimal investment exists,
see Theorem \ref{duda} in Section \ref{gorr} below. In this way we complement
earlier results of R\'asonyi (2016). Proofs will be given in Section \ref{proh},
while Section \ref{appo} contains auxiliary results.

\section{Setting}\label{kett}

For $x\in\mathbb{R}$ we denote $x^+:=\max\{0,x\}$ and $x^-:=\max\{0,-x\}$,
the positive/negative parts of $x$.
Let $(\Omega,\mathcal{F},P)$ be a probability space. The expectation of a random variable $X$ with respect
to some measure $Q$ on $(\Omega,\mathcal{F})$ is denoted by $E_Q[X]$. If $Q=P$ we drop the
subscript. $L^p$ is the set of
$p$-integrable random variables with respect to $P$, for $p\geq 1$. 
%If $Q=P$ we simply write $L^p$ instead of $L^p(Q)$ and 
We use $L^{\infty}$ to denote the family of
of (essentially) bounded random variables (with respect to $P$). When $X\in L^2$,
we define its variance as
\begin{equation}
\mathrm{var}(X):=E\left[(X-EX)^2\right].
\end{equation}
The symbol $\sim$ denotes
equivalence of measures. 

We now briefly describe the Arbitrage Pricing Model, concentrating on the mathematical
aspects. We assume that the economy in consideration contains assets with returns
\begin{eqnarray}
R_0 &:=& r;\quad R_i:=\mu_i+\bar{\beta}_i\varepsilon_i,\quad 1\leq i\leq m;\\
R_i &:=& \mu_i+\sum_{j=1}^m \beta^j_i\varepsilon_j+\bar{\beta}_i
\varepsilon_i,\quad i>m,
\end{eqnarray}
where the $\varepsilon_i$ are random variables and $\mu_i,\beta_i,\overline{\beta}_i$ are
constants. Asset $0$ represents a riskless investment with a constant rate of return $r\in\mathbb{R}$.
For simplicity, we set $r=0$ from now on.

The random variables $R_i$, $i\geq 1$ represent the profit (or loss) created tomorrow from
investing one dollar's worth of asset $i$ today. The random
variables $\varepsilon_i$, $i=1,\ldots,m$ serve
as \emph{factors} which influence the return on all the assets $i\geq 1$ while $\varepsilon_i,\ {i>m}$
are random sources particular to the individual assets $R_i$, $i>m$. 
We assume that the $\varepsilon_i$ are square-integrable, independent random variables satisfying 
\begin{eqnarray}
E[\varepsilon_i]=0,\quad E\left[\varepsilon_i^2\right]=1,\quad i\geq 1.
\end{eqnarray}
We further assume 
$\bar{\beta}_i\neq 0$, $i\geq 1$.

\begin{remark}
{\rm We deviate from the original definitions of Ross (1976) in two ways. First, in that
paper the $\varepsilon_i$ are assumed only uncorrelated. It seems, however, that
this condition is too weak to obtain interesting results in the present context,
see Proposition 4 of R\'asonyi (2004). Hence, as in Kabanov \& Kramkov (1998), R\'asonyi (2004), we take
independent $\varepsilon_i$. Second, in Ross (1976) it was not assumed that
the factor assets $R_i$, $1\leq i\leq m$ are tradeable. This is a rather mild additional
assumption we require, just like in Kabanov \& Kramkov (1998), R\'asonyi (2004).} 
\end{remark}

A {\em portfolio} 
$\psi$ in the assets $0,\ldots,k$ is
an arbitrary sequence $\psi_i,\ 0\leq i\leq k$ 
of real numbers satisfying 
\begin{equation}\label{lb}
\sum_{i=0}^k \psi_i=0.
\end{equation}

This means that we invest $\psi_i$ dollars in the respective assets today in such a way that
our investments add up to our initial capital $0$ (an arbitrary initial capital could
be treated similarly, we chose $0$ for simplicity again). Such a portfolio will have value
\begin{equation}
V({\psi}):=\sum_{i=0}^{k} \psi_i R_i
\end{equation}
tomorrow.

As $R_0=0$, each portfolio is characterized by $\psi_1,\ldots,\psi_k$. 
We reparametrize the model by introducing 
\begin{eqnarray}
b_i &:=& -\frac{\mu_i}{\bar{\beta}_i},\quad 1\leq i\leq m;\\
b_i &:=& -\frac{\mu_i}{\bar{\beta}_i}+
\sum_{j=1}^m \frac{\mu_j\beta^j_i}{\bar{\beta}_j\bar{\beta_i}},\quad i>m.
\end{eqnarray}
Asset returns then take the following form:
\begin{eqnarray}
R_i &=& \bar{\beta}_i(\varepsilon_i-b_i),\quad 1\leq i\leq m;\\
R_i &=& \sum_{j=1}^m\beta_i^j(\varepsilon_j-b_j)+\bar{\beta}_i(\varepsilon_i-
b_i),\quad i>m.
\end{eqnarray}

It is easy to see that the set
\begin{equation}
J_1:=\{V(\psi):\psi\mbox{ is a portfolio in assets }R_0,\ldots,R_k\}
\end{equation}
coincides with 
\begin{equation}
J_2:=\left\{\sum_{i=1}^k \phi_i (\varepsilon_i-b_i):\phi_1,\ldots,\phi_k\in\mathbb{R}\right\}.
\end{equation}
This motivates us to define \emph{elementary strategies}
$\mathcal{E}$ as the set of sequences $\phi=\{\phi_i,\ i\geq 1\}$ of real numbers 
such that, for some $k\geq 1$, $\phi_i=0$, $i> k$.
For each $\phi\in\mathcal{E}$, we set
\begin{equation}\label{tarra}
V(\phi):=\sum_{i=1}^{\infty} \phi_i (\varepsilon_i-b_i).
\end{equation}

This is a slight abuse of notation but it is justified by the equality $J_1=J_2$. From
now on we will use $V(\phi)$ exclusively in the sense \eqref{tarra}.

We interpret $\mathcal{E}$ as the set of possible portfolios using only finitely many
assets. For $\phi\in\mathcal{E}$, the value tomorrow is given by $V(\phi)$.
Define $K:=\{V(\phi):\phi\in\mathcal{E}\}$. 

This article is about optimization over the set of portfolios. 
As $K$ fails to be closed in any reasonable sense, we have no hope for finding an optimiser in it. 
Hence a natural idea is to enlarge the set of admissible portfolios in the
way presented below. 

\begin{assumption}\label{b}
We have 
\begin{equation}
\sum_{i=1}^{\infty} b_i^2<\infty.
\end{equation}
\end{assumption}

As explained in R\'asonyi (2004) and in Section 3 of R\'asonyi (2016), Assumption \ref{b} is equivalent to absence of
particular type of asymptotic arbitrage (in the sense of Ross (1976), Huberman (1982)) hence it is a plausible hypothesis on the
market parameters. 

Let $\ell_2:=\{\phi_i:\sum_{i=1}^{\infty}\phi_i^2<\infty\}$ be the family of square-summable
sequences. Recall that $\ell_2$ is a Hilbert space with the norm $||\phi||_{\ell_2}:=\sqrt{\sum_{i=1}^{\infty}\phi_i^2}$.
We set $\mathcal{A}:=\ell_2$, this will be the set of strategies over
which we will formulate our problem of utility maximization. 
Under Assumption \ref{b}, we define $V(\phi):=\sum_{i=1}^{\infty}\phi(\varepsilon_i-b_i)$,
where the series converges in $L^2$. 

Now set $K_1:=\{V(\phi):\phi\in\mathcal{A}\}$. 
Taking a supremum over $K_1$, it will be possible
to find a maximizer inside the same class, under appropriate assumptions, see Theorems
\ref{old_death} and \ref{duda} below.

\section{Existence of optimal strategies}\label{gorr}

Let $u:\mathbb{R}\to\mathbb{R}$ a concave, nondecreasing function.
We wish to prove that there is a strategy $\phi^*\in\mathcal{A}'(u)$ such that
\begin{equation}\label{proba}
E\left[u(V(\phi^*))\right]=\sup_{\phi\in\mathcal{A}'(u)}E\left[u(V(\phi))\right], 
\end{equation}
where 
\begin{equation}
\mathcal{A}'(u):=\{\phi\in\mathcal{A}:E\left[u^-(V(\phi))\right]<\infty\}
\end{equation}
is the set of strategies over which the optimization problem \eqref{proba} makes sense.

This is the problem of optimal investment with a risk-averse agent maximizing his/her 
expected utility from terminal portfolio wealth, see Chapters 2 and 3 of F\"ollmer \& Schied (2002) for
a detailed discussion in the case of finitely many assets.

We recall the results obtained in R\'asonyi (2016). The crucial hypothesis in that paper was the following.

\begin{assumption}\label{relevant}
For each $x\geq 0$, both 
\begin{equation}\label{vavelgrof}
\inf_{i\geq 1}P(\varepsilon_i>x)>0\mbox{ and }
\inf_{i\geq 1}P(\varepsilon_i<-x)>0 
\end{equation}
hold. Furthermore, 
\begin{equation}\label{unint}
\sup_{i\in\mathbb{N}}E[\varepsilon_i^2 1_{\{|\varepsilon_i|\geq N\}}]\to 0,\ N\to\infty.
\end{equation}
\end{assumption}

We now present Theorem 4.7 of R\'asonyi (2016).

\begin{theorem}\label{old_death} Let $u:\mathbb{R}\to\mathbb{R}$ be concave and non-decreasing, satisfying
\begin{equation}
u(x)\leq C(x^{\alpha}+1)\mbox{ for all }x\geq 0,
\end{equation}
with some $C>0$, $0\leq \alpha<1$.
Under Assumptions \ref{b} and \ref{relevant}, there exists $\phi^*\in\mathcal{A}'(u)$ such that
\begin{equation}
E\left[u(V(\phi^*))\right]=\sup_{\phi\in\mathcal{A}'(u)}E\left[u(V(\phi))\right].\Box
\end{equation}
\end{theorem}

In the present article our purpose is to complement Theorem \ref{old_death} by proving the
existence of optimizers under alternative conditions.
Whereas \eqref{unint} in Assumption \ref{relevant} is rather mild, \eqref{vavelgrof} is somewhat
restrictive: it excludes e.g. the case where
all the $\varepsilon_i$ are bounded random variables. It would thus be desirable to drop condition 
\eqref{vavelgrof}.  In Theorem \ref{duda}
below we manage to do so, at the price of requiring more integrability on the $\varepsilon_i$ than \eqref{unint}.
 
\begin{assumption}\label{novum}
There exists $\gamma>0$ such that
\begin{equation}\label{subgauss}
\sup_i E\left[e^{\gamma|\varepsilon_i|}\right]<\infty, 
\end{equation}
and, for all $i$,
\begin{equation}\label{villi}
P(\varepsilon_i<b_i),P(\varepsilon_i>b_i)>0.
\end{equation}
\end{assumption}

\begin{remark}
{\rm An arbitrage strategy (in the sense of e.g. F\"ollmer \& Schied (2002)) 
is $\hat{\phi}\in\mathcal{E}$ such that $V(\hat{\phi})\geq 0$ a.s. and $P(V(\hat{\phi})>0)>0$.
We say that the no-arbitrage condition (NA) holds in the given market if there are no arbitrage
strategies. (NA) is necessary for the existence of a maximizer for \eqref{proba} when $u$ is
strictly increasing. Indeed, no $\phi^*$ can be optimal if $\hat{\phi}$ violates (NA):
as $V(\phi^*+\hat{\phi})\geq V(\phi^*)$ and there is a strict inequality with positive probability, the strategy $\phi^*+\hat{\phi}$ outperforms $\phi^*$.

Condition \eqref{villi} is easily seen to imply (NA) in our setting: for any $\phi\in\mathcal{E}$
which is not identically zero,
$P(\sum_{i=1}^N \phi_i(\varepsilon_i-b_i)<0)\geq \prod_{i\in L_+(\phi)} P(\varepsilon_i<b_i)\prod_{i\in L_-(\phi)}
P(\varepsilon_i>b_i)>0$
since at least one of the two sets 
$L_+(\phi):=\{1\leq i\leq N:\phi_i> 0\}$, $L_-(\phi):=\{1\leq i\leq N:\phi_i<0\}$ is non-empty.
Conversely, if \eqref{villi} fails then (NA) is obviously violated. This shows that in our
market model \eqref{villi} is precisely the (NA) condition.} 
\end{remark}

The main result of this article is stated now.

\begin{theorem}\label{duda} Let $u:\mathbb{R}\to\mathbb{R}$ be concave and non-decreasing, satisfying
\begin{eqnarray}\label{motyo}
u(x) &\leq& C_1(x^{\alpha}+1)\mbox{ for all }x\geq 0,\\
\label{moko} u(x) &\leq& C_2(-|x|^{\beta}+1)\mbox{ for all }x<0
\end{eqnarray}
with some $C_1,C_2>0$, $0\leq \alpha<1<\beta$.
Under Assumptions \ref{b} and \ref{novum}, there exists $\phi^*\in\mathcal{A}'(u)$ such that
\begin{equation}
E\left[u(V(\phi^*))\right]=\sup_{\phi\in\mathcal{A}'(u)}E\left[ u(V(\phi))\right].
\end{equation}
\end{theorem}

\begin{remark}
{\rm Note that neither of Assumptions \ref{novum}, \ref{relevant} implies the other:
 \eqref{vavelgrof} is stronger than \eqref{villi} while \eqref{unint} is weaker than \eqref{subgauss}. Hence 
Theorem \ref{duda} complements Theorem \ref{old_death}. 
Notice also condition \eqref{moko} in Theorem \ref{duda} that is absent in Theorem \ref{old_death}.
As $u$ is concave and non-decreasing, \eqref{moko} always holds with $\beta=1$. Imposing $\beta>1$ is then only
a mild extra requirement.}
\end{remark}

\section{Proofs}\label{proh}

Condition \eqref{subgauss} in Assumption \ref{novum} implies that the family $\{V(\phi):\phi\in\mathcal{E},||\phi||_{\ell_2}\leq \delta\}$
enjoys a very strong, exponential uniform integrability property, for $\delta>0$ small enough, see
the next lemma and the ensuing remark.

\begin{lemma}\label{ui} Under condition \eqref{subgauss}, there is $\delta_0>0$ such that
\begin{equation}
\sup_{\phi\in\mathcal{E}, ||\phi||_{\ell_2}\leq\delta_0}Ee^{ |\sum_{i=1}^{\infty}\phi_i\varepsilon_i|}
\end{equation}
is finite.
\end{lemma}
\begin{proof}  
For each $i$ and for all $t$ with $|t|\leq\gamma/2$ there is a random variable 
$\xi_i(t)$ between $0$ and $t$ such that
\begin{equation}
e^{t\varepsilon_i}=1+t\varepsilon_i+t^2\varepsilon_i^2e^{\xi_i(t)\varepsilon_i}.
\end{equation}
It follows that
\begin{equation}
E\left[e^{t\varepsilon_i}\right]\leq 1+t^2 E\left[\varepsilon_i^2e^{\gamma|\varepsilon_i|/2}\right].
\end{equation}
By elementary properties of the logarithm and by \eqref{subgauss}, 
\begin{equation}
|\ln E\left[e^{t\varepsilon_i}\right]|\leq c t^2 E\left[\varepsilon_i^2e^{\gamma |\varepsilon_i|/2}\right]\leq C t^2,
\end{equation}
for $t$ small enough, say, for $|t|\leq T$, with some constants $c,C,T>0$ 
that are independent of $i$. If $\phi_i=0$ for $i>N$ then we then have
\begin{equation}
E\left[e^{ \sum_{i=1}^{\infty}\phi_i \varepsilon_i}\right]=\prod_{i=1}^N  E\left[e^{\phi_i\varepsilon_i}\right]
\leq e^{C\sum_{i=1}^N \phi_i^2}
\end{equation}
provided that $|\phi_i|\leq \min\{T,\gamma/2\}$ for all $i$. This implies, by $e^{|x|}\leq e^x+e^{-x}$,
\begin{equation}
E\left[e^{ |\sum_{i=1}^{\infty}\phi_i \varepsilon_i|}\right]\leq
2e^{C\sum_{i=1}^N \phi_i^2}=2e^{C\Vert\phi\Vert_{\ell_2}^2},
\end{equation}
for all $\phi\in\mathcal{E}$ with $\Vert\phi\Vert_{\ell_2}\leq \min\{T,\gamma/2\}$,
so one may set $\delta_0:=\min\{T,\gamma/2\}$.
\end{proof}

\begin{remark}\label{euler}
{\rm Lemma \ref{ui} implies that the family $\{e^{ |\sum_{i=0}^{\infty}\phi_i\varepsilon_i|}:\phi\in\mathcal{E}, ||\phi||_{\ell_2}\leq\delta\}$
is uniformly integrable for all $\delta<\delta_0$, a fortiori, the family 
$\{|\sum_{i=0}^{\infty}\phi_i\varepsilon_i|^2:\phi\in\mathcal{E}, ||\phi||_{\ell_2}\leq 1\}$ is also uniformly integrable.
It is this latter conclusion that we need in subsequent arguments. We wonder whether this could be deduced from an integrability
condition weaker than \eqref{subgauss}.}
\end{remark}

We denote by $\overline{K}$ the closure of $K$ with respect to convergence in probability.

\begin{lemma}\label{tartu}
Let Assumption \ref{b} and \eqref{subgauss} in Assumption \ref{novum} be in force. 
Then $K_1=\overline{K}$ so the set $K_1$ is closed in probability.
\end{lemma}
\begin{proof} We have $\mathcal{E}\subset\mathcal{A}$ and for every $\phi\in\mathcal{A}$,
 $V(\phi(n))\to V(\phi)$ in $L^2$ and hence also in probability, where $\phi(n)\in\mathcal{E}$ is
 such that $\phi_i(n)=\phi_i$,
 $i\leq n$, $\phi_i(n)=0$, $i>n$. Thus $K$ is dense in $K_1$ and it is enough to prove that $K_1$ is
 closed in probability.
 
Let $\phi(n)\in\mathcal{A}$ such that $V(\phi(n))\to X$ a.s. for some random variable $X$.
We may and will suppose that $\phi(n)\in\mathcal{E}$ for all $n$. First let us consider the case where
$\sup_n ||\phi(n)||_{\ell_2}=\infty$. By extracting a subsequence (which we continue to denote by $n$)
we may and will assume $||\phi(n)||_{\ell_2}\to\infty$, $n\to\infty$.
Define $\tilde{\phi}_i(n):=
\phi_i(n)/||\phi(n)||_{\ell_2}$ for all $n,i$. Clearly, $\tilde{\phi}(n)\in\mathcal{A}$
with $||\tilde{\phi}(n)||_{\ell_2}=1$ and
\begin{equation}
\lim_{n\to\infty}V(\tilde{\phi}(n))=0\mbox{ a.s.} 
\end{equation}
Let $M:=\sqrt{\sum_{i=1}^{\infty} b_i^2}$.
We obviously have $|\sum_{i=1}^{\infty} \tilde{\phi}_i(n)b_i|\leq M$ for all $n$.

By Remark \ref{euler}, the family $|V(\tilde{\phi}(n))|^2$, $n\geq 1$ is uniformly
integrable which implies that $V(\tilde{\phi}(n))\to 0$ in $L^2$ as well. But this is
absurd since $\mathrm{var}(V(\tilde{\phi}(n)))=1$ for all $n\geq 1$. This contradiction
shows that necessarily $\sup_n ||\phi(n)||_{\ell_2}<\infty$. 

Then there is a 
subsequence which is weakly convergent in $\ell_2$ and, by the Banach-Saks theorem applied in $\ell_2$, 
suitable convex combinations 
\begin{equation}
\widehat{\phi}(N):=\frac{1}{N}\sum_{k=1}^N \phi_{n_k}
\end{equation}
with a subsequence $n_k$ satisfy
\begin{equation}
\left\Vert\widehat{\phi}(N)-\phi^*\right\Vert_{\ell_2}^2=\sum_{i=1}^{\infty} \left(\widehat{\phi}_i(N)-
\phi_i^*\right)^2\to 0,\ N\to\infty,
\end{equation}
for some $\phi^*\in\mathcal{A}=\ell_2$. Hence, by orthonormality of the system 
$\varepsilon_i$, $i\geq 1$
in $L^2$, 
\begin{equation}
E\left[\left(V(\widehat\phi(N))-V(\phi^*)\right)^2\right]\to 0,\quad N\to\infty,
\end{equation}
so $V\left(\widehat\phi(N)\right)\to V(\phi^*)$ in probability as well.
As $\widehat{\phi}(N)$ are convex combinations of a subsequence of $\phi(n)$,
$V\left(\widehat\phi(N)\right)\to X$ in probability. This implies
$V(\phi^*)=X$ and the statement of this lemma follows.
\end{proof} 

Let $\mathcal{M}$ denote the set of $Q\sim P$ such that $E_QR_i=0$ for all $i\geq 1$.
Elements of $\mathcal{M}$ are the \emph{risk-neutral measures} for this market model.
It is not a priori clear that $\mathcal{M}\neq \emptyset$.

\begin{remark}\label{nulla} {\rm We recall that if $Q\in\mathcal{M}$ is such that $dQ/dP\in L^2$ then
$E_Q\left[V(\phi)\right]=0$ for all $\phi\in\mathcal{A}$, see Lemma 3.4 of R\'asonyi (2016).}
\end{remark}

The next ingredient for the proof of Theorem \ref{duda} is the fact that there are 
$Q\in\mathcal{M}$ with suitable
integrability properties. 

\begin{lemma}\label{uff} Assume that
\begin{equation}\label{harom}
\sup_{i\geq 1} E\left[|\varepsilon_i|^3\right]<\infty
\end{equation}
and  Assumption \ref{b} is in force.
Then for all $p\geq 1$ there exists $Q=Q(p)\in\mathcal{M}$ such that $dQ/dP,dP/dQ\in L^p$.
\end{lemma}
\begin{proof}
This result essentially follows from Theorem 2 in R\'asonyi (2004).
%(see also Theorem 7.3.3 in \cite{phd}). 
We recall the main elements of that argument.
Define the function $\psi(x):=1/2+1/(1+e^x)$, $x\in\mathbb{R}$. Using the
implicit function theorem, it is shown in the proof of Theorem
2 in R\'asonyi (2004) that there exists a sequence $a_i$, $i\geq 1$ and $N^*\in\mathbb{N}$
such that for all $N\geq N^*$, 
\begin{equation}
\frac{dQ(N)}{dP}:=\prod_{i=N}^{\infty} \frac{\psi(a_i(\varepsilon_i-b_i))}{E\psi(a_i(\varepsilon_i-b_i))}
\end{equation}
defines $Q(N)\sim P$ such that $E_{Q(N)}\left[\varepsilon_i\right]=b_i$ for $i\geq N$. 
Furthermore, there is a constant $L>0$
such that 
\begin{equation}\label{ab}
|a_i|\leq L|b_i|\mbox{ for }i\geq N^*.
\end{equation}

It can now be checked by direct calculations that, for each $w\in\mathbb{R}$,
the function
\begin{equation}
(a,b)\to f_i(a,b):=\ln \frac{E\left[\psi^w(a(\varepsilon_i-b))\right]}
{\left(\left[E\psi(a(\varepsilon_i-b))\right]\right)^w}
\end{equation}
is twice continuosly differentiable in a neighbourhood $U$ of $(0,0)$ (which does not depend on $i$),
$f_i(0,0)=\partial_a f_i(0,0)=\partial_b f_i(0,0)=0$ and the second derivatives are uniformly (in $i$)
bounded in $U$.  This fact was used in Theorem 2 of R\'asonyi (2004) for $w=2$ but it holds, in fact,
for every $w$.

There is $N'\geq N^*$ such that $(a_i,b_i)\in U$ for $i\geq N'$, hence
\begin{equation}
E\left[\left(\frac{dQ(N')}{dP}\right)^w\right]\leq \prod_{i=N'}^{\infty} e^{c\sum_{i=1}^{\infty} (a_i^2+b_i^2)}
\end{equation}
for a constant $c>0$. Noting \eqref{ab} and Assumption \ref{b} we conclude that $(dQ(N')/dP)^w\in L^1$. 
It is clear
that the argument also works for two values $w=\pm p$ at the same time and this provides $Q(N')$
such that $dQ(N')/dP,dP/dQ(N')\in L^p$. 

For $1\leq i<N'$, let $g_i:=dW/dP$ where $W$ is given by Corollary \ref{meggy} for the choice $X:=\varepsilon_i-b_i$.
By independence of the $\varepsilon_i$, it is clear that 
\begin{equation}
\frac{dQ}{dP}:=\frac{dQ(N')}{dP}\prod_{i=1}^{N'-1} g_i
\end{equation}
provides a suitable probability $Q$.
\end{proof}

\begin{lemma}\label{fontos} 
Let $u:\mathbb{R}\to\mathbb{R}$ be concave and non-decreasing, satisfying
\begin{eqnarray}\label{harsh}
u(x)\leq C_1(x^{\alpha}+1)\mbox{ for all }x\geq 0,\\
 u(x)\leq C_2(-|x|^{\beta}+1)\mbox{ for all }x<0,
\end{eqnarray}
with some $C_1,C_2>0$, $0\leq \alpha<1<\beta$. If there exists 
$Q\in\mathcal{M}$ such that $dQ/dP\in L^q$, $dP/dQ\in L^r$, where $q\geq\max\{2,\beta/(\beta-1)\}$, 
$r>\alpha/(1-\alpha)$ then
there exists $X\in \overline{K}_1$ such that
\begin{equation}\label{benign}
E\left[u(X)\right]=\sup_{\phi\in\mathcal{A}'(u)}E\left[u(V(\phi))\right],
\end{equation}
where $\overline{K}_1$ denotes the closure of $K_1$ for convergence in probability.
\end{lemma}
\begin{proof}
The arguments in this proof are inspired by those of Theorem 4.7 in R\'asonyi (2016). Denote $Y:=V(\phi)$ for some 
$\phi\in\mathcal{A}'(u)$. By H\"older's inequality, Remark \ref{nulla} and \eqref{harsh},
\begin{eqnarray}\nonumber
E\left[u(Y^+)\right] &\leq& C_1\left(E\left[(Y^+)^{\alpha}\right]+1\right)\leq 
C_1\left(C'\left( E_Q \left[Y^+\right]\right)^\alpha +1\right)=\\ 
C_1\left(C'\left(E_Q\left[Y^-\right]\right)^\alpha +1\right)
&\leq& C_1\left(C'C''\left(E\left[(Y^-)^{\beta}\right]\right)^{\alpha/\beta} +1\right)\nonumber \\ &\leq&
C_1\left(C'C''\left(-(1/C_2)E\left[u(-Y^-)\right]+1\right)^{\alpha/\beta}+1\right)\label{diehard}
\end{eqnarray}
with $C':=\left(E\left[(dP/dQ)^{\alpha/(1-\alpha)}\right]\right)^{1-\alpha}$, 
$C'':=\left(E\left[(dQ/dP)^{\beta/(\beta-1)}\right]\right)^{\alpha(\beta-1)/\beta}$. 
Let $\phi_n\in\mathcal{A}'(u)$ be such that $E\left[u(V(\phi_n))\right]\to \sup_{\phi\in
\mathcal{A}'(u)}E\left[u(V(\phi))\right]$, $n\to\infty$. If we had $\sup_n E_Q\left[|V(\phi_n)|\right]=\infty$ then
(along a subsequence), $E\left[u(-V^-(\phi_n))\right]\to -\infty$, $n\to\infty$, by \eqref{diehard}.
Using the fact that for a non-increasing function $u$ one has 
\begin{equation}
u(x)\leq u(x^+)+u(-x^-)+|u(0)|,
\end{equation}
it follows that
\begin{equation}
E\left[u(V(\phi_n))\right]\leq E\left[u(V^+(\phi_n))\right]+E\left[u(-V^-(\phi_n))\right]+|u(0)|\to-\infty,
\end{equation}
by \eqref{diehard} and by $\alpha/\beta<1$, which is absurd. Hence $\sup_n E_Q\left[|V(\phi_n)|\right]$ 
must be finite and Koml\'os's theorem
(see Koml\'os (1967)) implies that, for some convex combinations of a subsequence,
\begin{equation}
\widehat{\phi}_N=\frac{1}{N}\sum_{k=1}^N \phi_{n_k},
\end{equation}
one has $V(\widehat{\phi}_N)\to X$ a.s., $N\to\infty$.
By convexity of $K_1$, clearly $X\in \overline{K}_1$. We claim that $X$ satisfies \eqref{benign}.
 
Indeed, choose $\theta>\alpha$ such that $\theta/(1-\theta)= r$. We get by H\"older's inequality,
\begin{eqnarray}
E\left[(V^+(\widehat{\phi}_N))^{\theta}\right] &\leq& C'''\left(E_Q\left[V^+(\widehat{\phi}_N)\right]\right)^\theta
\leq C'''\left(\sup_n E_Q\left[ |V(\phi_n)|\right] \right)^{\theta}
\end{eqnarray}
with constant $C''':=\left(E\left[(dP/dQ)^{\theta/(1-\theta)}\right]\right)^{1-\theta}$, for all $n$,
recalling that the $\widehat{\phi}_N$ are convex combinations of the $\phi_n$. As the latter expression
has a finite supremum in $n$ by our arguments above, we get by \eqref{harsh} that the family $u(V^+(\widehat{\phi}_N))$,
$N\in\mathbb{N}$ is uniformly integrable. This, combined with Fatou's lemma and the convexity of $u$, implies 
\begin{equation}
E\left[u(X)\right]\geq \limsup_{N\to\infty} E\left[u(V(\widehat{\phi}_N))\right]\geq 
\limsup_{n\to\infty} E\left[u(V({\phi}_n))\right],
\end{equation}
and the proof is finished.
\end{proof}

\begin{proof}[Proof of Theorem \ref{duda}.] Notice that Assumption \ref{novum} implies
\eqref{harom}. Choose $p$ with $p>\max\{2,\beta/(\beta-1),\alpha/(1-\alpha)\}$ and invoke Lemma \ref{uff}.
From Lemma \ref{fontos} we get a maximizer $X\in \overline{K}_1$ and, by Lemma \ref{tartu},
$X=V(\phi^*)$ for some $\phi^*\in\mathcal{A}$. Since $0\in\mathcal{A}'(u)$, $E\left[u(V(\phi^*))\right]
\geq E\left[u(0)\right]>-\infty$ 
so $\phi^*\in\mathcal{A}'(u)$ as well and the proof is finished.
\end{proof}

\begin{remark} {\rm
It is clear from the above proofs that $E[u(X^+)]<\infty$ hence 
\begin{equation}
E\left[u(V(\phi^*))\right]=\sup_{\phi\in\mathcal{A}'(u)}E\left[ u(V(\phi))\right]<\infty.
\end{equation}}
\end{remark}

\section{Appendix}\label{appo} 

Here we briefly treat utility maximization for the case of one asset, which is rather elementary but
indispensable for the proof of Theorem \ref{duda}. The arguments below are
standard, see e.g. Davis (1997), F\"ollmer \& Schied (2002), but we could not find a reference that would cover the
setting we need.

Let $X$ be a real-valued random variable with $P(X>0),P(X<0)>0$ and $E|X|<\infty$. 
Let furthermore $u$ be concave and nondecreasing, $u(x)= cx$, $x<0$, 
$u(x)\leq C(x^{\alpha}+1)$ for $x\geq 0$ with
some $c,C>0$, $0\leq\alpha<1$.

\begin{proposition}\label{jules}
There exists $\phi^*\in\mathbb{R}$ such that 
\begin{equation}
\sup_{\phi\in\mathbb{R}}E\left[u(\phi X)\right]=E\left[u(\phi^* X)\right].
\end{equation}
\end{proposition}
\begin{proof}
Let $\phi>0$ first. Then $E\left[u(\phi X)\right]\leq C(\phi^{\alpha}E\left[(X^+)^{\alpha}\right] +1)-\phi 
cE\left[X^-\right]$
and this tends to $-\infty$ as $\phi\to\infty$ since $E\left[X^-\right]>0$. A similar argument
shows that $E\left[u(\phi X)\right]\to -\infty$ as $\phi\to -\infty$. The function $\phi\to E\left[u(\phi)\right]$ is
continuous. To see this, let $\phi_n \to \phi$. Clearly, for $D:=\sup_n |\phi_n|+1$ we
have $|u(\phi_n X)|\leq \max\{c,C\}[(D|X|+1)+1]$ for all $n$, using the inequality $x^{\alpha}\leq x+1$.
Lebesgue's theorem implies
$E\left[u(\phi_n X)\right]\to E\left[u(\phi X)\right]$. Now the statement is clear.
\end{proof}

\begin{corollary}\label{meggy}
For each $p\geq 1$ there is $W=W(p)\sim P$ with $E_W[X]=0$ such that $dW/dP\in L^{\infty}$ and
$dP/dW\in L^p$.
\end{corollary}
\begin{proof} Define $u(x)=\alpha x$, $x\leq 0$ and $u(x)=(x+1)^\alpha-1$ for $x>0$ with
some $0<\alpha<1$. Clearly, $u$ is continuously differentiable.
We claim that $\phi\to E\left[u(\phi X)\right]$ is also continuously differentiable. 
It suffices to show that the set of random variables $u'(\phi X)X$, $\phi\in\mathbb{R}$ is dominated
by a random variable of finite expectation, which is clear as $|u'(\phi X)X|\leq \alpha |X|$.
By Proposition \ref{jules}, the function $\phi\to E[u(\phi X)]$ attains its maximum at $\phi^*$, so
$E\left[u'(\phi^* X)X\right]=0$ which implies $E_W[X]=0$ for the
probability $W\sim P$ defined by
\begin{equation}
\frac{dW}{dP}:=\frac{u'(\phi^* X)}{E\left[u'(\phi^* X)\right]}.
\end{equation}
As $u'$ is bounded, $dW/dP\in L^{\infty}$. We also have that $dP/dW$
is smaller than constant times $(|\phi^* X| +1)^{1-\alpha}$. Since $X$ has
a finite expectation, the latter expression is in $L^p$ for an arbitrarily large $p$
provided that $\alpha$ is chosen close enough to $1$.
\end{proof}

{\bf Acknowledgments.} I thank the referee for useful comments. This work was supported by the
Hungarian Academy of Sciences [``Lend\"ulet'' Grant LP2015-6].


\begin{thebibliography}{10}

\bibitem{sara}
S. Biagini \& M. Frittelli (2005) \newblock
Utility maximization in incomplete markets for unbounded processes,
\newblock\emph{Finance Stoch.} \textbf{9}, 493--517.

\bibitem{josef} C. Cuchiero, I. Klein \& J. Teichmann (2015)
\newblock A new perspective on the fundamental theorem of 
asset pricing for large financial markets, 
\newblock {\em Teoriya Veroyatnostei i ee Primeneniya},
\textbf{60}, 660--686. 
\newblock {\em Forthcoming in Theory of Probability and its Applications}, 2016.\newblock\texttt{arXiv:1412.7562}

\bibitem{davis}
M.H.A. Davis (1997)
\newblock Option pricing in incomplete markets,
\newblock In: {\em Mathematics of derivative securities} (M. A. H. Dempster and S. R. Pliska, ed.), 
216--226, Cambridge University Press.

\bibitem{paolo} M. De Donno, P. Guasoni \& M. Pratelli (2005)
\newblock Superreplication and utility maximization in large financial markets.
\newblock\emph{Stochastic Process. Appl.} \textbf{115}, 2006--2022.

\bibitem{fs}
H.~F{\"o}llmer \& A.~Schied (2002)
\newblock \emph{Stochastic finance: an introduction in discrete time.}
\newblock Walter de Gruyter \& Co., Berlin.

\bibitem{huang} Ch.-F. Huang \& R. H. Litzenberger (1988)
\newblock{\em Foundations for financial economics.} \newblock Elsevier.

\bibitem{huberman}
G. Huberman (1982) 
\newblock A simple approach to arbitrage pricing theory,
\newblock \emph{J. Econom. Theory}  \textbf{28}, 289--297.


\bibitem{kabanov-kramkov}
Yu. M. Kabanov \& D. O. Kramkov (1994)
\newblock Large financial markets: asymptotic arbitrage and contiguity,
\newblock {\em Theory Probab. Appl.} \textbf{39}, 182--187.

\bibitem{kk}
Yu. M. Kabanov \& D. O. Kramkov (1998) 
\newblock Asymptotic arbitrage in large financial markets,
\newblock \emph{Finance Stoch.} \textbf{2}, 143--172.


\bibitem{ks_exp}
Yu. M. Kabanov \& Ch. Stricker (2002)
\newblock On the optimal portfolio for the exponential utility maximization: remarks to the six-author paper,
\newblock \emph{Math. Finance} \textbf{12}, 125--134.


\bibitem{klein}
I. Klein (2000)
\newblock A fundamental theorem of asset pricing for large financial markets,
\newblock {\em Math. Finance}, \textbf{10}, 443--458.

\bibitem{klein-schachermayer1}
I. Klein \& W. Schachermayer (1996a)
\newblock Asymptotic arbitrage in non-complete large financial markets,
\newblock {\em Theory Probab. Appl.} \textbf{41}, 780--788.

\bibitem{klein-schachermayer2}
I. Klein \& W. Schachermayer (1996b)
\newblock A quantitative and a dual version of the {H}almos-{S}avage theorem
  with applications to mathematical finance,
\newblock {\em Ann. Probab.} \textbf{24}, 867--881.

\bibitem{komlos}
J.~Koml{\'o}s (1967)
\newblock A generalization of a problem of {S}teinhaus,
\newblock {\em Acta Math. Acad. Sci. Hungar.} \textbf{18}, 217--229.


\bibitem{KS99}
D.~O. Kramkov \& W.~Schachermayer (1999)
\newblock The asymptotic elasticity of utility functions and optimal investment
  in incomplete markets,
\newblock \emph{Ann. Appl. Probab.} \textbf{9}, 904--950.

\bibitem{owen}
M.~P. Owen \& G.~{\v{Z}}itkovi{\'c} (2009)
\newblock Optimal investment with an unbounded random endowment and
  utility-based pricing,
\newblock {\em Math. Finance} \textbf{19}, 129--159.

\bibitem{def} M. R\'asonyi (2004)
\newblock Arbitrage pricing theory and risk-neutral
measures,
\newblock {\em Decis. Econ. Finance} \textbf{27}, 109--123.

\bibitem{callum} M. R\'asonyi (2016) \newblock Maximizing expected utility in the Arbitrage
Pricing Model, \newblock{\em Submitted}, \newblock\texttt{arXiv:1508.07761v2} 


\bibitem{ross}
S. A. Ross. (1976)
\newblock The arbitrage theory of capital asset pricing,
\newblock \emph{J. Econom. Theory} \textbf{13}, 341--360.


\bibitem{rokhlin} D. B. Rokhlin (2008) 
\newblock {Asymptotic arbitrage and num\'eraire portfolios in large financial markets},
\newblock \emph{Finance Stoch.} \textbf{12}, 173--194.



\bibitem{w}
W.~Schachermayer (2001)
\newblock Optimal investment in incomplete markets when wealth may become
  negative,
\newblock \emph{Ann. Appl. Probab.} \textbf{11}, 694--734.


\end{thebibliography}
\end{document}